\documentclass[]{interact}
\usepackage{hyperref}
\usepackage{epstopdf}
\usepackage[caption=false]{subfig}
\usepackage[doublespacing]{setspace}

\usepackage[longnamesfirst,sort]{natbib}
\setcitestyle{authoryear,open={(},close={)}}

\theoremstyle{plain}
\newtheorem{theorem}{Theorem}[section]

\theoremstyle{definition}

\theoremstyle{remark}

\begin{document}

\articletype{ARTICLE TEMPLATE}

\title{A correlation structure for the analysis of Gaussian and non-Gaussian responses in crossover experimental designs with repeated measures}

\author{
\name{N.A. Cruz\textsuperscript{a}\thanks{N.A. Cruz . Email: neacruzgu@unal.edu.co, Corresponding author}, O.O. Melo \textsuperscript{b} and C.A. Martinez \textsuperscript{c}}
\affil{\textsuperscript{a}Phd Student, Department of  Statistics, 
             Faculty of Sciences,
           Universidad Nacional de Colombia, Phone:(+571)\;3165000 Ext: 13206,\; Fax:(+571)\;3165000\; Ext: 13210 ; \textsuperscript{b}oomelom@unal.edu.co, Associate Professor, Department of  Statistics, 
             Faculty of Sciences,
            Universidad Nacional de Colombia; \textsuperscript{c}cmartinez@agrosavia.co,  Phd, Associate Researcher,
 Corporación Colombiana de Investigación Agropecuaria – AGROSAVIA, Sede Central}
}

\maketitle

\begin{abstract}
  In this study, we propose a family of correlation structures for crossover designs with repeated measures for both, Gaussian and non-Gaussian responses using generalized
estimating equations (GEE). The structure considers two matrices: one that models between-period correlation and another one that models within-period correlation. The overall correlation matrix, which is used to build the GEE, corresponds to the Kronecker between these matrices. A procedure to estimate the parameters of the correlation matrix is proposed, its statistical properties are studied and a comparison with standard models using a single correlation matrix is carried out. A simulation study showed a superior performance of the proposed structure in terms of the quasi-likelihood criterion, efficiency, and the capacity to explain complex correlation phenomena/patterns in longitudinal data from crossover designs
\end{abstract}

\section{Introduction}
Experimental designs are a very useful tool to analyze the effects of treatments applied to a set of experimental units \citep{hk}. In this kind of studies, it is frequent that experimental units are observed at a unique time point, this is known as a transversal experiment; notwithstanding, sometimes experimental units are observed several times during the study, keeping the same treatment, giving rise to longitudinal studies \citep{Davis}. 
There are situations in which an experimental unit receives all treatments, each one in a different period, this induces a setup where a sequence of treatments is applied to each unit. This kind of designs is known as crossover design \citep{ken15}.
In the scope of crossover designs, published results focus on the case of a normally distributed   \citep{ken15} or binary (two possible outputs, namely success or failure) response variable. The latter case has been treated by means of generalized linear models for binary data (\cite{rat92}, \cite{curtin2017meta} and \cite{li2018power}).\\
\citet[pag 204]{ken15} described a crossover experiment with three treatments to control arterial pressure: treatment A is a placebo, treatments B and C are 20 and 40 mg doses of a test drug. Thus, there were six three-period sequences: ABC, ACB, BCA, BAC, CAB, and CBA, each one of them was applied to two individuals. In each period, 10 consecutive measurements of diastolic arterial pressure were taken: 30 and 15 minutes before, and 15, 30, 45, 60, 75, 90, 120 and 240 minutes after the administration of the treatment, as shown in Table \ref{tabla1}. 
\begin{table}[ht]
    \centering
   \begin{tabular}{c|c|c|c|}
        Sequence & Period 1 & Period 2 & Period 3\\
        \hline
        \begin{tabular}{cc}
            (1) ABC & Ind 1\\
             & Ind 2
        \end{tabular} & \begin{tabular}{c}
               10 measurements\\
              10 measurements
        \end{tabular}&  \begin{tabular}{c}
              10 measurements\\
              10 measurements
        \end{tabular}&  \begin{tabular}{c}
              10 measurements\\
             10 measurements
        \end{tabular}\\
        \vdots & \vdots &\vdots &\vdots \\
         \begin{tabular}{cc}
            (6) CBA & Ind 11\\
             & Ind 12
        \end{tabular} & \begin{tabular}{c}
               10 measurements\\
              10 measurements
        \end{tabular}&  \begin{tabular}{c}
              10 measurements\\
              10 measurements
        \end{tabular}&  \begin{tabular}{c}
              10 measurements\\
             10 measurements
        \end{tabular}\\
        \hline
    \end{tabular}
    \caption{Structure of the blood pressure crossover design}
    \label{tabla1}
\end{table}
. A second experiment carried out by researchers of the Department of Animal Sciences of Universidad Nacional de Colombia, was focused on inferring the effects of two diets, A (grass) and B (a mixture of grass and paper recycling waste) on dairy cattle performance. Eight cows split into two groups of four received the diets in such a way that the first group was fed diet A from day 1 to day 42, and diet B from day 43 to day 84, while the second group was fed diet B during the first period and diet A during the second one. Measurements of milk yield and quality, live weight and body condition score were taken at 1, 14, 28, 42, 56, 70 and 84 days \citep{jaime}; the design structure is shown in Table\ref{tablavacas}.

\begin{table}[ht]
    \centering
   \begin{tabular}{c|c|c|}
        Sequence & Period 1 & Period 2\\
        \hline
        \begin{tabular}{cc}
            (1) AB & Ind 1\\
            &\vdots\\
             & Ind 4
        \end{tabular} & \begin{tabular}{c}
               3 measurements\\
               \vdots \\
              3 measurements
        \end{tabular}&  \begin{tabular}{c}
             3 measurements\\
               \vdots \\
              3 measurements
        \end{tabular}\\
         \begin{tabular}{cc}
             (1) BA & Ind 5\\
            &\vdots\\
             & Ind 8
        \end{tabular} & \begin{tabular}{c}
                3 measurements\\
               \vdots \\
              3 measurements
        \end{tabular}&  \begin{tabular}{c}
              3 measurements\\
               \vdots \\
              3 measurements
        \end{tabular}\\
        \hline
    \end{tabular}
    \caption{Structure of the crossover design in cows}
    \label{tablavacas}
\end{table}

To analyze this sort of designs, \cite{basu2010joint}, \cite{josephy2015within}, \cite{hao2015explicit}, \cite{lui2015test}, \cite{rosenkranz2015analysis}, \cite{grayling2018blinded}, \cite{madeyski2018effect} and \cite{kitchenham2018corrections} used mixed models for crossover designs with Gaussian response and a single observation per period including additive carryover effects. \cite{biabani2018crossover} presented a review on crossover designs in neuroscience, all papers considered Gaussian responses and did not account for carryover effects due to the presence of a washout period, even when it lasted a few days. On the other hand,  \cite{oh2003bayesian},  \cite{curtin2017meta} and   \cite{li2018power} used generalized linear models for crossover designs of two periods and two sequences of two treatments with a continuous (normal or gamma) or binary response variable and each experimental unit observed once per period. They used generalized estimating equations (GEE) to estimate model parameters. A Bayesian formulation of generalized linear models where the response was the survival time with a single observation per period was presented in  \cite{shkedy2005hierarchical} and \cite{liubayesian}.

Moreover, \cite{dubois2011model}, \cite{diaz2013random} and \cite{for15} used Gaussian mixed models to analyze records from crossover designs with repeated measures using the area under curve as a strategy to obtain a single observation per period and one experimental unit, and they did not account for carryover effects.

In all aforementioned approaches, it is assumed that the crossover design considers a washout period between treatments and that the response variable of each individual is observed once per period. These assumptions are not fulfilled in the experiments described above because of two reasons: i) in patients with arterial hypertension the treatment cannot be stopped, while in the case of dairy cattle, cows must be fed every day; moreover, in both studies, the placebo is part of the treatment design, so considering it as a washout period would modify the experiment, ii) in each treatment, experimental units were observed several times per period. Therefore, there is a necessity for developing a consistent methodology to analyze this sort of experiments. 
In this study, we develop a methodology to analyze data from crossover designs with repeated measures using GEE and considering two correlation structures, between and within periods, which are combined via a Kronecker product to yield the overall correlation matrix. These approach was found to improve the estimation of parameters of interest with respect to methods based on GEE that consider a single correlation structure for all the observations of a subject; in addition, this method accounts for carryover effects, hence, it does not need a washout period.

In the second section, we present some background and define the methodology, in the third section we propose a method to estimate the correlation matrix, and provide theoretical results that support our estimation and modelling approaches. In the fourth section, we present a simulation study showing some advantages of our model as compared to the model with a single correlation structure for each period. Lastly, in the fifth section, the methodology is applied to real data from the aforementioned arterial pressure and dairy cattle experiments. 
 
\section{Crossover designs with repeated measures}
A crossover design has the following components \citep{pat51}: i) \textit{Sequences} which are each one of the distinct treatment combinations to be sequentially and randomly applied to the experimental units,  ii) \textit{Treatments}, which are randomly applied to each experimental unit within each sequence, iii) \textit{Periods}, the time intervals in which the treatments that make up a sequence are applied, usually, each period is the same for all sequences and, consequently, the number of periods is equal to the number of elements of each sequence,  iv) \textit{Experimental unit}, the subjects or elements that receive the treatments, in each sequence there are  $n_l$ experimental units, so the total number of experimental units in the study is denoted by  $n=\sum_{l=1}^S n_l$.

Another important feature of the crossover design is the existence of carryover effects, defined in  \cite{vegas2016crossover} as follows: persistency of the effect of a treatment over those that will be applied later, i.e., the treatment is applied in a given period, but its effect still affects the response in later periods when other treatments are applied, this residual effect is known as carryover effect. When the effect persists after one period, it is known as first order carryover effect, and when it persists after two periods, it is called a second order effect and so on \citep{pat51}. 

 In a crossover design with $S$ sequences of length (the number of elements comprising each sequence)  $P$, $Y_{ijk}$ is defined as the  $k$th respose of the $i$th experimental unit at $j$th period. Let $n_{ij}$ be the number of observations of the $i$th experimental unit during the  $j$th period. Then, vector $\pmb{Y}_{ij}$ is defined as: 
 \begin{equation}\label{observacion}
  \pmb{Y}_{ij}=\left(Y_{ij1}, \ldots, Y_{ijn_{ij}} \right)^T
 \end{equation}
Also, vector $\pmb{Y}_{i}$ is defined as:
 \begin{equation}
\pmb{Y}_{i}= \left(\pmb{Y}_{i1}, \ldots, \pmb{Y}_{iP}\right)^T
 \end{equation}
which has dimension $\sum_{j=1}^P n_{ij}$, where $P$ is the number of periods.\\
According to these definitions, the ideas presented when discussing the arterial pressure and dairy cattle experiments, and assuming that  $Y_{ijk}$ has a distribution in the exponential family, we propose the following model based on GEE:
\begin{align}
       E(Y_{ijk})&=\mu_{ijk},\qquad i=1, \ldots, n, \; j=1,\ldots, P, \;  k=1, \ldots, n_{ij}, \; L=\max_{ij}\{n_{ij}\} \nonumber\\
   g(\mu_{ijk})&=\pmb{x}^T_{ijk}\pmb{\beta}=\mu +\gamma_j + \tau_{d[i,j]}+ \theta_{d[i,j-u]}+\cdots+\theta_{d[i,j-q]} \label{predic_eta}
\end{align}
where $g(\cdot)$ is the link function related to the exponential family, $\pmb{x}_{ijk}$ is the vector of the design matrix corresponding to $k$th response from the $i$th experimental unit at the  $j$th period, $\pmb{\beta}$ is the vectors of fixed effects, $\mu$ is the overall mean, $\alpha_i$  is the effect of the $i$th sequence, $\gamma_j$ is the effect of the $j$th period, $ \tau_{d[i,j]}$ is the effect of treatment $d$ applied in the period $j$ to $i$th experimental unit ($d=1, \ldots, q$),  $\theta_{(1)[i,j-1]}$ is the first order carryover effect of $d$ treatment,  $\theta_{d[i,j-u]}$  is the carry over effect of order $u$ of $d$ treatment. The carryover effects are considered in the two experiments discussed above because there was not a washout period. Due to the fact that observations of the same experimental unit are correlated, parameter estimation is carried out using GEE \citep{liang1986longitudinal}. To this end, the following system of $q=dim(\pmb{\beta})$ equations has to be solved:
 
\begin{align}
\pmb{U}(\pmb{\beta})=&\left[ \left\{\sum_{i=1}^{n_i} \pmb{x}_{mi}^T \pmb{D}\left( \frac{\partial \pmb{\mu}_i}{\partial\eta}\right) [\pmb{V}(\pmb{\mu}_i)]^{-1}\left(\frac{\pmb{y}_{i}- \pmb{\mu}_i}{a(\phi)} \right) \right\}_{m=1, \ldots, q} \right]_{q\times 1}\label{ec101}
\end{align}
where $\pmb{\mu}_i=(\mu_{i11}, \ldots, \mu_{iPn_{iP}})^T$, $\pmb{y}_{i}=(y_{i11}, \ldots, y_{iPn_{iP}})^T$, $\pmb{x}_{mi}$ is the $m$th column of the design matrix of the ith experimental unit, and the covariance component  $\pmb{V}(\pmb{\mu}_i)$ is defined as:
\begin{equation}\label{ec104}
\pmb{V}(\pmb{\mu}_i)=\left[\pmb{D}(V(\mu_{ijk})^{\frac{1}{2}}) \pmb{R}(\pmb{\pmb{ \alpha}}) \pmb{D}({V}(\mu_{ijk})^{\frac{1}{2}})\right]_{P \times P} 
\end{equation}
where $\pmb{D}(\cdot)$ is a diagonal matrix,  $V(\mu_{ijk})$ is the variance function corresponding to the exponential family, and  $R(\pmb{ \alpha})$ is the correlation matrix related to the covariance matrix $\pmb{\Sigma}_{i}=Var(\pmb{Y}_{i})$ as follows:
\begin{align*}
    \left( \pmb{\Sigma}_{i}\right)_{LP \times LP}&= \left( \pmb{D}(Var(Y_{ijk})^{\frac{1}{2}}) \pmb{R}(\pmb{\alpha}) \pmb{D}(Var(Y_{ijk})^{\frac{1}{2}})\right)_{LP \times LP}.
 \end{align*}
 
\section{Kronecker correlation matrix}
Due to the repeated measures structure in the crossover design, we propose a correlation structure of the form:
\begin{equation}\label{ec105}
    \pmb{R}(\pmb{\alpha})=\pmb{\Psi}\otimes\pmb{R}_1(\pmb{\alpha}_1)
\end{equation}
where $\pmb{R}_1(\pmb{\alpha}_1)$  is the within-period correlation matrix, $\pmb{\Psi}$ is the between-period correlation matrix, and $\otimes$ represents the Kronecker product \citep{harville}.
To estimate this matrix, we propose the following modified GEE to estimate $\pmb{\beta}$ is:
\begin{align}\nonumber
\pmb{U}_1(\pmb{\beta})=&\left[ \left\{ \sum_{i=1}^{n} \pmb{x}_{mi}^T \pmb{D}\left( \frac{\partial \pmb{\mu}_i}{\partial\eta}\right) [\pmb{V}(\pmb{\mu}_i)]^{-1}\left(\frac{\pmb{y}_{i}- \pmb{\mu}_i}{a(\phi)} \right) \right\}_{m=1, \ldots, p} \right]_{p\times 1}\label{ec101}
\end{align}
where $\pmb{V}(\pmb{\mu}_i)$ is defined as in Equation \eqref{ec104}, $\pmb{\mu}_i=\{\mu_{i11}, \ldots, \mu_{iPL}\}$, $\pmb{x}_{mi}$ is the $m$th column of the design matrix of the $i$th experimental unit.\\
The estimating equation for  $\pmb{\alpha}_1$ is:
\begin{equation}\label{106}
\pmb{U}_2(\pmb{ \alpha}_1)=\sum_{i=1}^n \left(\frac{\partial \pmb{\varepsilon}_{i}}{\partial \pmb{ \alpha}_1} \right)^T \pmb{H}_{i}^{-1} \left(\pmb{W}_{i} - \pmb{\varepsilon}_{i} \right)
\end{equation}
where $\pmb{H}_{i}=\pmb{D}(V(r_{ijk}))_{q\times q}$ is a diagonal matrix, $\pmb{\varepsilon}_{i}=E(\pmb{W}_{i})_{q\times 1}$,\\ 
$\pmb{W}_{i}= (r_{i11}r_{i12},$ $r_{i11}r_{i13},\ldots,r_{iP(L-1)}r_{iPL})^T_{q\times 1}$, and $r_{ijk}$ is the  $ijk$th observed Pearson residual defined as:
\begin{equation}\label{ecPearson_obs}
    r_{ijk}=\frac{Y_{ijk}-\hat{\mu}_{ijk}}{\hat{\phi}\sqrt{V(\hat{\mu}_{ijk})})}
\end{equation}
and $q={P\choose 2}$.
On the other hand, we propose the following estimator for $\pmb{\Psi}=\left(\psi_{jj'}\right)_{P\times P}$:
\begin{equation}\label{107A}
\hat{\psi}_{jj'}=\frac{1}{n} \sum_{i=1}^n tr\left(\pmb{R}_1(\pmb{\hat{\alpha}}_1)(\pmb{r}_{(j)i}-\bar{\pmb{r}}_{(j)})(\pmb{r}_{(j')i}-\bar{\pmb{r}}_{(j')})^T\right)
\end{equation}
where $\pmb{r}_{(j)i}$ is the vector of Pearson residuals for the $i$th experminetal unit at the $j$th  period, and $\bar{\pmb{r}}_{(j)}$ is the average (over i),  $j\neq j'$, with $j= 1, \ldots, P$\\
\begin{theorem}\label{teorema2}
The estimator of $\pmb{R}(\pmb{\alpha})$ given by $\pmb{\hat{\Psi}}\otimes\pmb{R}_1(\pmb{\hat{\alpha}}_1)$ is asymptotically unbiased and consistent. 
\end{theorem} 
\begin{proof}
See appendix \ref{ApenAA}
\end{proof}
The form of  $\pmb{R}_1(\pmb{\alpha}_1)$ in (\ref{ec105}) is given by correlation structures such as: independence, autoregressive, exchangeable, etc \citep{Davis}.\\

The correlation structures are compared via the quasi-likelihood information criterion ($QIC$, \cite{panq}), which is defined as:
\begin{equation}\label{QIC}
    QIC=-2QL(\hat{\mu};\pmb{I})+2trace(\pmb{\hat{\Omega}}^{-1}_I\pmb{\hat{V}_R})
\end{equation}
where $\hat{\mu}_{ijk}=\hat{\eta}_{ijk}=g^{-1}(\pmb{x}_{ijk}\hat{\beta})$ is the estimated expected value of observation $Y_{ijk}$  under the model assuming the correlation matrix $R$, $\hat{\Omega}_I$ is the estimated covariance matrix of $\pmb{\beta}$ under independence, and  $\hat{V}_R$ is the covariance matrix of  $\pmb{\beta}$ under the model with correlation matrix   $\pmb{R}$ defined as in \eqref{ec105}.
\section{Simulation Study} 
\begin{figure}[h!]
 \centering 
 \scalebox{0.45} 
 {\includegraphics{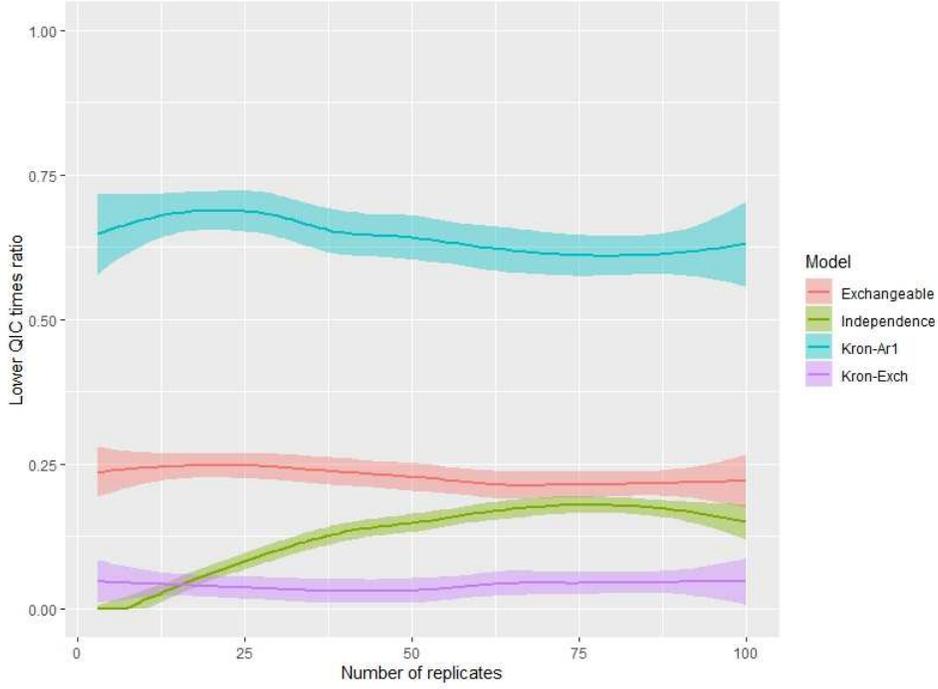}}
 \caption{Mean and corresponding 95\% confidence interval (computed from the 100 simulations) for the proportion of the number of times that the $QIC$ selected each model according to the number of replicates per sequence.}
 \label{fig3} 
\end{figure}
A simulation study was carried out to evaluate the proposed model. The setup was a crossover model with three periods  $(P=3)$, five repetitions  $(L=5)$ and two sequences $(S=2)$ within each period. The number of replicates per sequence  $(n)$ varied from 2 to 100,  and for each replicate, the simulation was ran 100 times. The following parameters are used in the simulation:
\begin{align}
   & P_k \sim N(0, 1), \; k=1,2,3, \nonumber \\
   & T_j \sim N(0, 1),\; j=1,2,3,4,5,\nonumber\\
   &S_{jk} \sim N(0,1),\nonumber\\
   &\Psi=\{\psi_{ab}\}_{3\times 3},\; \psi_{ab}\sim U(-1,1) ,\; \Psi\geq 0,\nonumber\\
   &\pmb{R}_1(\alpha_1), \;r_{ab}=\alpha_1^{|a-b|},\; \alpha_1\sim U(-1,1),\; \pmb{R}_1(\alpha_1)\geq 0,\label{ecsim}\\
   &  \mu_{ijk}=\alpha + P_k+T_j+S_{(ij)}, \; i=1, \ldots, 3n,\nonumber\\
    & Y_{ijk}\sim N(\mu_{ijk}, \sigma^2),\nonumber\\
    & \pmb{Y}_i= \left(\pmb{Y}_{i1}, \ldots, \pmb{Y}_{iP}\right)^T,\nonumber\\
    & Corr(\pmb{Y}_i)=\pmb{\Psi} \otimes \pmb{R}_1(\alpha_1)\nonumber
\end{align}
Four GEE models were fit per replicate, all had the same linear predictor and differed in the correlation structure: 1) independence, 2) first order autoregressive, 3) exchangeable, and 4) the Kronecker matrix proposed in this paper. Both, the simulations and fitted model were carried out using the libraries \textbf{gee} \citep{gee} and \textbf{geeM} \citep{geeM} of R \citep{RRR}

\begin{figure}[h!]
 \centering 
 \scalebox{0.7} 
 {\includegraphics{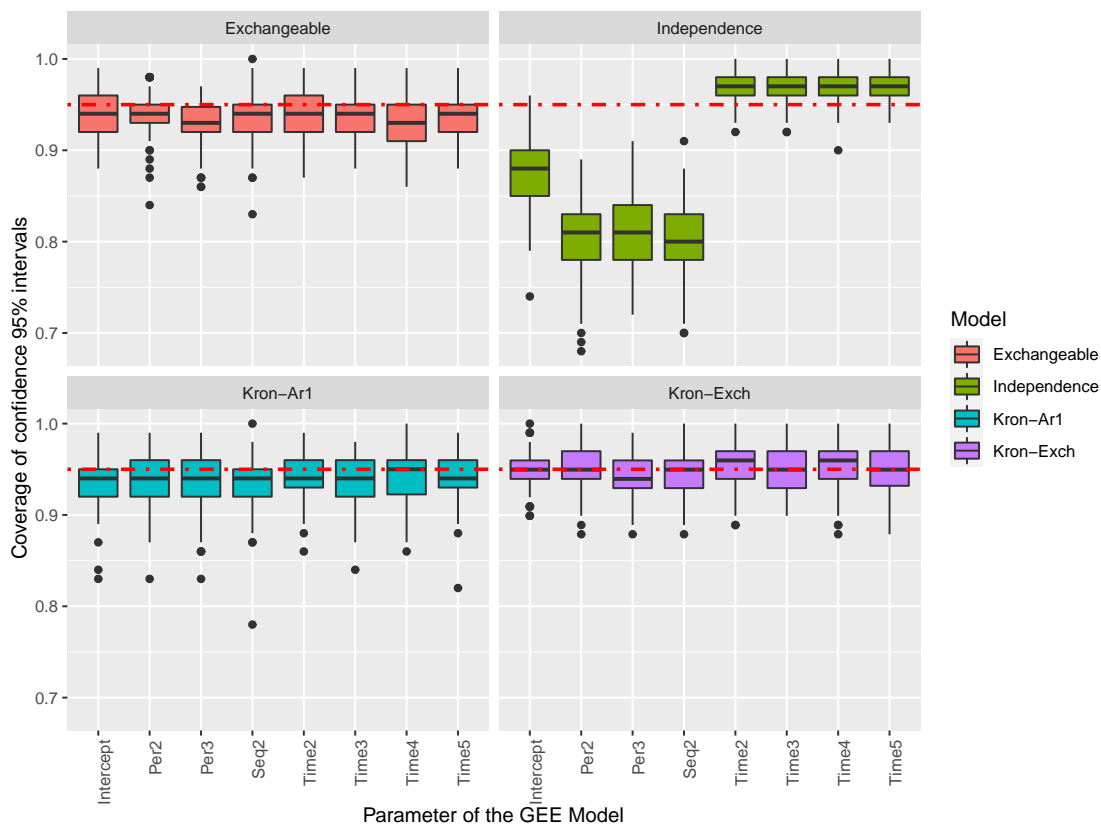}} 
 \caption{Box plots of 95\% confidence intervals coverage for each location parameter of the GEE models with different correlation structures. The red line represents the expected coverage.}
 \label{fig4} 
\end{figure}

Figure \ref{fig3} shows the proportion of times that each model had the smallest QIC relative to the 100 simulations for each number of replications per sequence (the mean and 95\% confidence interval are shown). Moreover, an analysis of parameter estimation under each one of these models using the simulated data was performed; Figure 2 presents 95\% confidence intervals for the location parameters under each model. In the simulation study it is observed that most of the times the lowest $QIC$ is obtained for the model with the true correlation structure. This suggets that the proposed estimation method along with theorem \ref{teorema2}, which will allow to detect kronecker correlation structures in the data of the crossover design. In addition, this behavior is maintained even when the number of replicas within the design is low.\\
Regarding the coverage of confidence intervals for each parameter, in Figure \ref{fig4} it can be noticed that the model with independence structure has a very low coverage for period and sequence effects; this will not allow a correct inference about the parameters of interest in the design. The other three models had more appropriate coverages, but the model with the within-period xchangeable matrix (that is, the true structure) exhibited values closer to 95\%.
As to the time effects, corresponding coverages were close to 95\% in all four models; however, a subcoverage was observed for the model with a within-period exchangeable structure and an overcoverage for the independence model. The true model showed coverages corresponding to 95\%. 
\section{Application}
\begin{figure}[ht]
\centering 
\includegraphics[width=15cm]{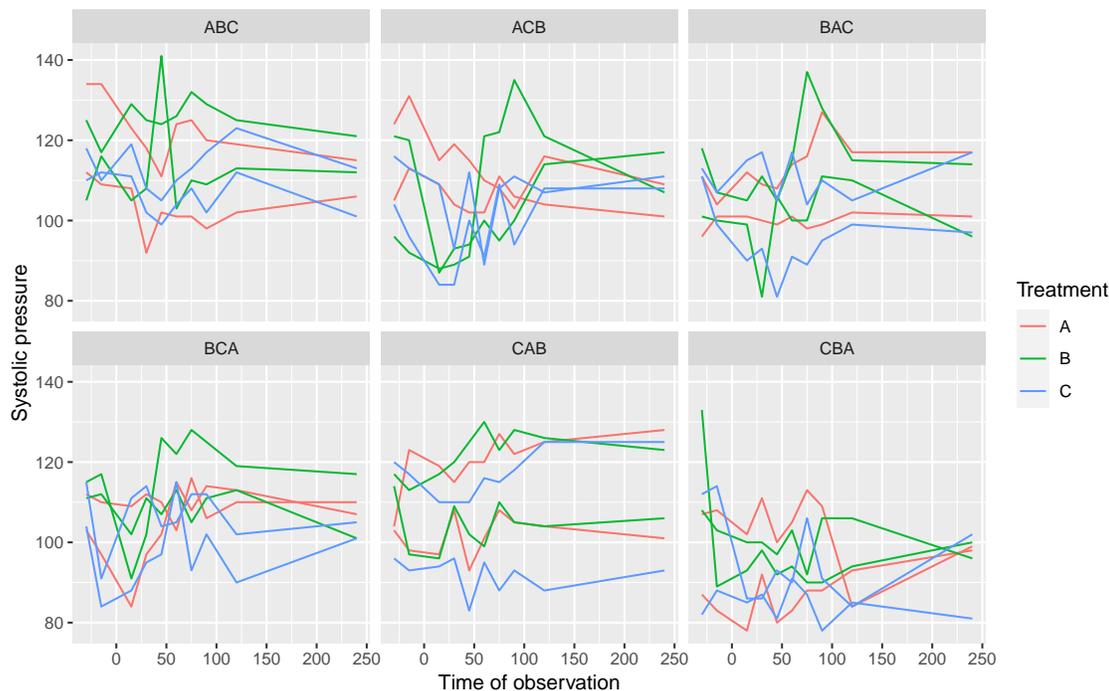}
\caption{Blood systolic pressure (mmHg), the time of observation is presented in minutes from the application}
\label{datosreales}
\end{figure}
The systolic blood pressure data is presented in Figure \ref{datosreales}. The model used to analyze this experiment had a linear predictor considering fixed effects of treatment, period, baseline (the two measurements taken before applying the treatment), first and second order linear and quadratic carryover effects as a function of time. On the other hand, the model use to analyze data from the second experiment had a linear predictor considering fixed effects of baseline, treatment, period, and first order linear carryover effects as a function of time (quadratic effects were not considered because there were three observations within each period).  The structures of the working correlation matrices were: i) independence, ii) first order autoregressive, and iii) fitted individually and were compared via the $QIC$ to determine the one exhibiting the best fit to each dataset. Table \ref{tabla2} presents the $QIC$ for each correlation structure in both datasets. 
\begin{table}[ht]
    \centering
   \captionsetup{justification=centering,margin=2cm}
   \begin{tabular}{|c|c|c|}\hline
   Matriz $R(\alpha)$ & $QIC$ presión arterial& $QIC$ vacas\\
   \hline
   $\pmb{I}_{PL}$ & 46086& 200.65\\
   \hline
   $\pmb{I}_{P}\otimes AR(1)_{L}$ & 44166& 139.70 \\
   \hline
   $\pmb{I}_{P}\otimes Exch_{L}$ & 45059 &138.85 \\
   \hline
    $AR(1)_{PL}$ & 44923 & 143.85\\
    \hline
    $Exch_{PL}$ &  44158 & 138.27 \\
   \hline
   $\pmb{\Psi}\otimes AR(1)_{L}$ & 43400 & 180.34 \\
   \hline
   $\pmb{\Psi}\otimes Exch_{L}$ & 43393 & 138.22\\
   \hline
    \end{tabular}
    \caption{Correlation matrices and the corresponding $QIC$ values}
    \label{tabla2}
\end{table}

\begin{figure}[ht]
 \centering 
 \scalebox{0.7} 
 {\includegraphics{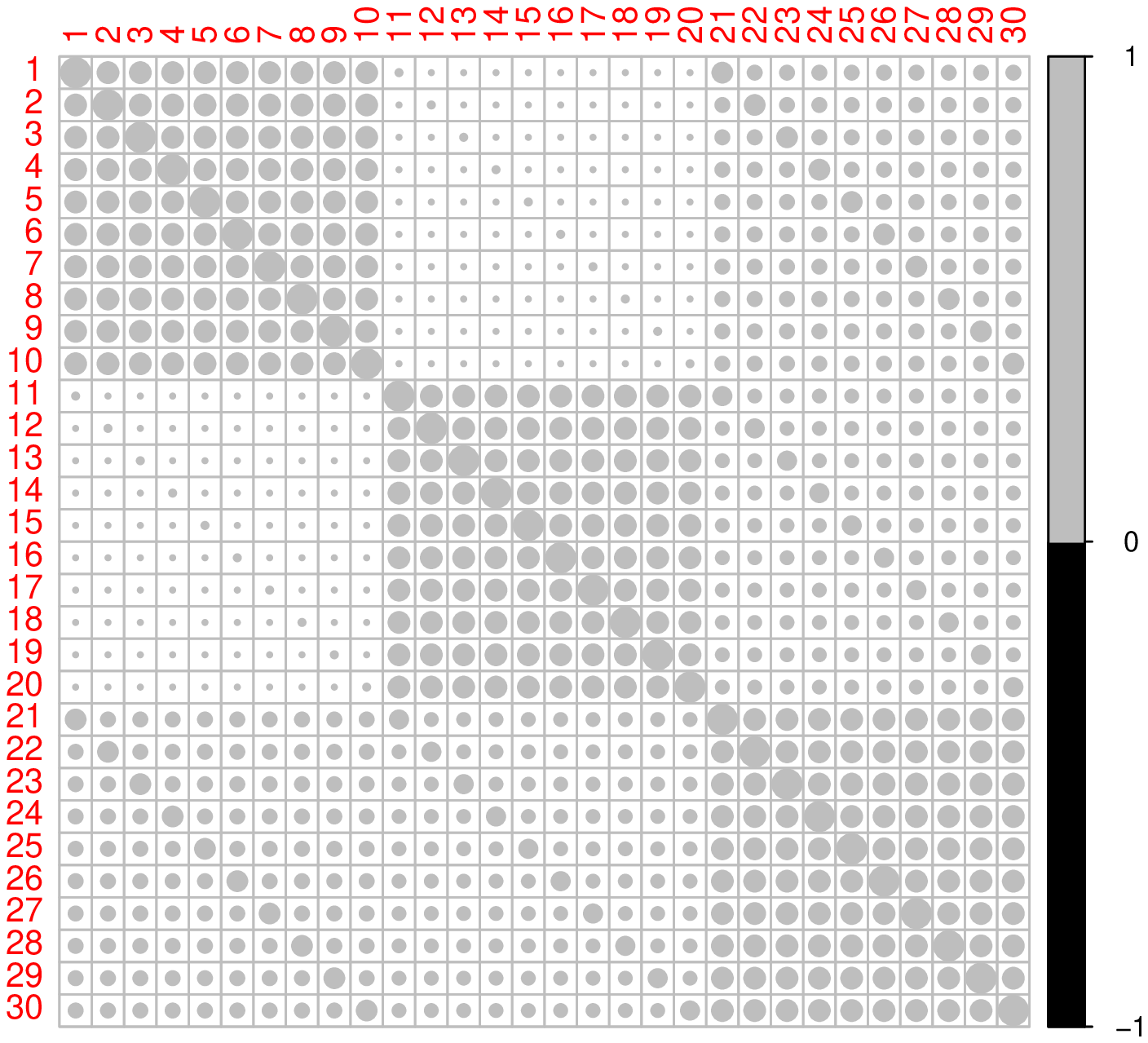}} 
 \caption{Estimated correlation matrix for the arterial pressure data with the structure $\pmb{\Psi}\otimes Exch_{10}$}
 \label{fig1} 
\end{figure}
\begin{figure}[ht]
 \centering 
 \scalebox{0.6} 
 {\includegraphics{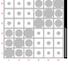}} 
 \caption{Estimated correlation matrix for the dairy cattle data with the structure  $\pmb{\Psi}\otimes Exch_{3}$}
 \label{fig11} 
\end{figure}
\begin{table}[ht]
\centering
\begin{tabular}{|l|cccc|}
  \hline
  Parameter &  Estimate  & Robust SE & Robust $z$ & $p$-value\\ 
  \hline
 Intercept &106.2917  & 3.4327 & 30.9641 & $<$0.001  \\  
 Base & 8.9587  & 1.9515 & 4.5908&$<$0.001\\  
 Time & 1.0754  & 0.4209 & 2.5548 &0.0106 \\
 Period 2 & 5.9752  & 3.0859 & 1.9363 & 0.0528\\  
 Period 3 &10.4117  & 5.1565 & 2.0191 & 0.0435 \\ 
 Treatment B &-1.3032  & 2.0353 & -0.6403 &0.5220 \\ 
 Treatment C &  -10.8840  & 3.4096 & -3.1922 & 0.0014 \\ 
 Carry-over B &  -7.2221  & 4.0008 & -1.8052 & 0.0710\\ 
 Carry-over C & -12.1513  & 5.6162 & -2.1636 &0.0305 \\ 
 Time$^2$ & -0.0526  & 0.0202 & -2.6091 & 0.0091\\  
 Time$\times$ Carry over B & 0.9191  & 0.7102 & 1.2940& 0.1957\\  
 Time$\times$ Carry over C &  -0.1007  & 0.5789 & -0.1739 &0.8619\\ 
 Time$^2\times$ Carry over B & -0.0459  & 0.0405 & -1.1325 &0.2574\\ 
 Time$^2\times$ Carry over C &   0.0209 & 0.0410 & 0.5097 &0.6103\\
   \hline
\end{tabular}
\caption{Estimates obtained by GEE for arterial pressure data }
\label{tabla3}
\end{table}
\begin{table}[ht]
\centering
\begin{tabular}{|l|cccc|}
  \hline
  Parameter &  Estimate  & Robust SE & Robust $z$ & $p$-value\\ 
  \hline
 Intercept &19.5505  & 1.0248 & 10.2950 & $<$0.001  \\
 Base & 0.5850  & 0.0629 & 9.2961&$<$0.001\\  
 Time & -0.6293 & 0.0949 & -6.6321 &$<$0.001\\
 Period 2 & 2.3908 & 1.0395 & 2.3000 & 0.0214 \\  
  Treatment A &1.1217 & 0.8341 & 1.3448 & 0.1787\\ 
 Carry-over A &   -4.2480 & 1.4128 & -3.0069 & 0.0026\\ 
 Time$^2$ &  0.0117 & 0.0018 & 6.4550 &$<$0.001 \\  
 Time$\times$ Carry over A &-0.0702 & 0.0260 & -2.6949 & 0.0070 \\ 
   \hline
\end{tabular}
\caption{Estimates obtained by GEE for dairy cattle data }
\label{tabla33}
\end{table}
According to \ref{tabla2} the $\pmb{\Psi}\otimes Exch_{10}$ and $\pmb{\Psi}\otimes Exch_{3}$ correlation matrices had the smallest QIC values for the arterial pressure and dairy cattle experiments, respectively. Hence, correlation matrices having the Kronecker structure had the best fit in the two datasets. The correlation matrices estimated using equations \eqref{106} and \eqref{107A} for the arterial pressure data are, respectively:
\begin{align}
    \hat{\pmb{\Psi}}&=\begin{pmatrix}
1.0000 & 0.0537 & 0.4486 \\ 
 0.0537 & 1.0000 & 0.3756 \\ 
0.4486 & 0.3756 & 1.0000 \\ 
        \end{pmatrix} \label{sii2}\\
   \pmb{\hat{R}}_1(\pmb{\alpha}) &=\begin{pmatrix}
 1.0000 & 0.4958 & 0.4958 & \cdots& 0.4958 \\ 
  0.4958 & 1.0000 & 0.4958 &\cdots & 0.4958 \\ 
 \vdots & \vdots & \vdots & \ddots & \vdots \\ 
  0.4958 & 0.4958 & 0.4958 & \cdots& 1.0000 \\
        \end{pmatrix}\label{siii3}
\end{align}
and for the dairy cattle dataset:
\begin{align}
    \hat{\pmb{\Psi}}&=\begin{pmatrix}
 1.0000 & 0.1073 \\ 
 0.1073 & 1.0000 \\ 
        \end{pmatrix} \label{sii22}\\
   \pmb{\hat{R}}_1(\pmb{\alpha}) &= \begin{pmatrix}
1.0000 & 0.5610 & 0.5610 \\ 
0.5610 & 1.0000 & 0.5610 \\ 
0.5610 & 0.5610 & 1.0000 \\
        \end{pmatrix}\label{siii32}
\end{align}
Figure \ref{fig1} shows the matrix computed as the Kronecker product of matrices in equations  \eqref{sii2} and \eqref{siii3}. Notice the positive correlation between periods 1 and 3, a small positive correlation between periods 1 and 2. On the other hand, the matrix obtained as the Kronecker product of matrices in equations  \eqref{sii22} and \eqref{siii32}  is shown in Figure  \ref{fig11} where a positive but small correlation between periods 1 and 2 can be seen. \\

The matrices with the Kronecker structure are used to estimate the location parameters of the linear model (since these had the smallest QIC) yielding Table \ref{tabla3} 4 for the arterial pressure data and Table \ref{tabla33}  for the dairy cattle data. Each table shows the estimates, their standard errors (computed using the “sandwich” variance estimator \citep{Hardin}, the z statistic and the p-value corresponding to the null hypothesis $\beta=0$

Table \ref{tabla3} shows significant effects of baseline, treatment C (with respect to A), time (linear and quadratic), period, and carryover effect of treatment C. The interaction between carryover effects and time was not significant, which means that the effect of the previous treatment remains the same during the rest of the experiment. Moreover, Table \ref{tabla33} shows significant effects of baseline (milk yield before the beginning of the experiment), linear (negative) and quadratic (positive) regression coefficients of time, which means that milk yield shows a convex behavior during the experiment. There was no significant treatment effect, but there was a significant carryover effect of the paper-based diet over the grass-based diet, i.e., the cows being fed paper took a time to recover after changing to the grass-based diet.    

\section{Conclusions}

Defining the correlation matrix structure is a highly relevant decision when using GEE, a proper specification of correlation structures in longitudinal data analysis improves estimation efficiency, leading to more reliable statistical inferences \citep{hin2009working} . Hence, in this paper we develop a family of correlation structures that combine some of the classical structures used in longitudinal data analysis with an additional matrix that allows more flexibility in the overall correlation structure by adding a few parameters. Moreover, we provide an explicit estimation method of these parameters which features some sound statistical properties.  

The theoretical results supporting some asymptotic properties of the proposed estimators are illustrated through the simulation study where a gain in goodness of fit was observed across the different simulation scenarios.  The QIC was able to select the correct model, which had the proposed correlation structure. In addition, the confidence intervals built from the GEE showed a coverage that matched the nominal value. The estimation by intervals for each of the parameters presents coverage close to 95\%, which shows a correct theoretical specification of each univariate interval.

As to the real data analysis, the results for the arterial pressure data showed the importance of accounting for carryover effects as they are useful for correctly estimating and interpreting the treatment affects across the time. If not included in the model, these residual effects may induce confusion problems. On the other hand, in the dairy cattle data, significant carryover effects were detected as well. Regarding the estimates of correlation matrices in both datasets, the QIC selected the proposed correlation matrix. 
Thus, the theoretical and empirical results from real and simulated data analyses suggest that the proposed methodology is promising and may be applied to perform better inferences from data obtained under crossover designs without washout periods and a repeated measures structure. 

\newpage
\bibliography{references}
\begin{appendix}
\section{}\label{ApenAA}

\begin{theorem}
The estimator of $\pmb{R}(\pmb{\alpha})$ given by $\pmb{\hat{\Psi}}\otimes\pmb{R}_1(\pmb{\hat{\alpha}}_1)$ is asymptotically unbiased and consistent. 
\end{theorem} 

\begin{proof}
First, asymptotic properties of Pearson's residuals concerning expectation, variance, and residuals are explored.
Define the theoretical Pearson residuals of response of the $k$th response of the $i$th experimental unit at the $j$th period as:
\begin{equation}\label{eq7}
    R_{ijk}=\frac{Y_{ijk}-\mu_{ijk}}{{\phi}\sqrt{V(\mu_{ijk})})}
\end{equation}
and adapting the results of \citet{cox1968general} it is true that:
\begin{align}
    E(R_{ijk})&=\sum_{l=1}^pB(\hat{\beta}_l)E(H_l^{(ijk)})\nonumber\\
    &-\sum_{l,s=1}^p \mathcal{K}^{ls}E\left(H_l^{(ijk)}U_s^{(ijk)}+\frac{1}{2}H_{ls}^{(ijk)} \right)+ O(n^{-1})\label{exprijk}\\
    Var(R_{ijk})&=1+2\sum_{l=1}^pB(\hat{\beta_{l}})E(R_{ijk}H_l^{(ijk)})\nonumber\\
    &-\sum_{l,s=1}^p \mathcal{K}^{ls}E\left(2R_{ijk}H_l^{(ijk)}U_s^{(ijk)}+H_l^{(ijk)}H_s^{(ijk)} R_{ijk}H_{ls}^{(ijk)} \right)+ O(n^{-1}) \label{varrijk}
\end{align}
where
\[H_l^{(ijk)}=\frac{\partial R_{ijk}}{\partial\beta_l}, \; H_{ls}^{(ijk)}=\frac{\partial^2 R_{ijk}}{\partial\beta_l\partial\beta_s}\]
 and $\mathcal{K}^{ls}$  is the element at position $(l,s)$ of the inverse of the Fisher information matrix, and according to \citet{cordeiro1991bias}, the bias of $\hat{\pmb{\beta}}$ ($B(\hat{\pmb{\beta}})$) is given by:
\begin{align}
     B(\hat{\pmb{\beta}})&=-\frac{1}{2\phi}\pmb{(X^TWX)}\pmb{X}^T\pmb{D}(z_{ijk})\pmb{F}\pmb{1}\label{sesgobeta}\\
   \pmb{W}&=\pmb{D}(V(Y_{ijk})^{\frac{1}{2}}) \pmb{R}(\pmb{\alpha}) \pmb{D}(V(Y_{ijk})^{\frac{1}{2}}) \pmb{D}\left( \frac{\partial \pmb{\mu}_i}{\partial\eta}\right)\label{ecuacionww}\\
   \pmb{F}&=\pmb{D}\left(V(\mu_{ijk})^{-1}\left(\frac{\partial \mu_{ijk}}{\partial \eta_{ijk}}\right) \left(\frac{\partial ^2 \mu_{ijk}}{\partial \eta^2_{ijk}}\right)  \right)\nonumber
\end{align}
where $\pmb{1}$ is a vector of appropriate size whose entries are all equal to 1, $\pmb{D}(z_{ijk})$ is a diagonal matrix with elements on the diagonal given by the variance of the estimated linear predictors, i.e. , the diagonal of the matrix
\begin{equation}
           \pmb{z}=Var(\hat{\eta}_{111}. \ldots, \hat{\eta}_{nPL})
\end{equation}
i.e., $z_{ijk}=Var(\hat{\eta}_{ijk})$ and $\pmb{X}$ is the design matrix of the parametric effects described in Equation \eqref{predic_eta}. Now, computing the expected values by taking into account the properties of the exponential family we obtain:
\begin{align}
    E\left(H_l^{(ijk)}\right)&=-\sqrt{\phi V(\mu_{ijk})} \left(\frac{\partial \mu_{ijk}}{\partial \eta_{ijk}}\right) x_{l(ijk)} \label{ehl}\\
     E\left(H_{ls}^{(ijk)}\right)&=\left[2V(\mu_{ijk})^{-\frac{3}{2}} \left(\frac{\partial V(\mu_{ijk})}{\partial \mu_{ijk}}\right) \left(\frac{\partial \mu_{ijk}}{\partial \eta_{ijk}}\right)^2-
     2V(\mu_{ijk})^{-\frac{1}{2}}\left(\frac{\partial^2 \mu_{ijk}}{\partial \eta_{ijk}^2}\right) \right]\nonumber\\
     &\times \frac{1}{2}\sqrt{\phi}x_{l(ijk)}x_{s(ijk)}\label{ehls}\\
     E\left(H_l^{(ijk)}U_s^{(ijk)}\right)&=-\frac{1}{2}\phi^{\frac{1}{2}} V(\mu_{ijk})^{-\frac{3}{2}} \left(\frac{\partial V(\mu_{ijk})}{\partial \mu_{ijk}}\right) \left(\frac{\partial \mu_{ijk}}{\partial \eta_{ijk}}\right)^2x_{l(ijk)}x_{s(ijk)}\label{hlus}\\
     E\left(2R_{ijk} H_l^{(ijk)}U_s^{(ijk)}\right)&=-V(\mu_{ijk})^{-2}  \left(\frac{\partial V(\mu_{ijk})}{\partial \mu_{ijk}}\right)^2 \left(\frac{\partial \mu_{ijk}}{\partial \eta_{ijk}}\right)^2x_{l(ijk)}x_{s(ijk)} - \nonumber\\
       & 2\phi \left(\frac{\partial V(\mu_{ijk})}{\partial \mu_{ijk}}\right)^{-1} \left(\frac{\partial \mu_{ijk}}{\partial \eta_{ijk}}\right)^2 x_{l(ijk)}x_{s(ijk)}\label{2rhlus}\\
        E\left(H_l^{(ijk)}H_s^{(ijk)}\right)&= \left[ \phi + \frac{\left(\frac{\partial V(\mu_{ijk})}{\partial \mu_{ijk}}\right)}{4V(\mu_{ijk})}\right]w_{ijk} x_{l(ijk)}x_{s(ijk)}\label{hlhs}\\
        E\left(R_{ijk}H_{ls}^{(ijk)}\right)&= \frac{1}{2}\phi V(\mu_{ijk})^{-\frac{1}{2}} \left(\frac{\partial \mu_{ijk}}{\partial \eta_{ijk}}\right)\label{rhls}
     \end{align}
    where $w_{ijk}$ is the element $ijk$ of the diagonal of the matrix $\pmb{W}$ defined in Equation \eqref{ecuacionww}. From Equation \eqref{ehl} and the bias of $\pmb{\beta}$ given in Equation \eqref{sesgobeta}, it follows that:
    \begin{equation}\label{primera_esperanza}
     \sum_{l=1}^p B(\hat{\beta_{l}})E\left(H_{l}^{(ijk)}\right)=-\phi^{\frac{1}{2}} V(\mu_{ijk})^{-\frac{1}{2}} \left(\frac{\partial \mu_{ijk}}{\partial \eta_{ijk}}\right)\pmb{e}_{ijk}\pmb{X}B(\hat{\pmb{\beta}})
    \end{equation}
where $\pmb{e}_{ijk}$ is a vector of zeros with a 1's at the $ijk$th position. From  equations \eqref{ehls} and  \eqref{hlus}, we get:
\begin{align}
     E\left(H_l^{(ijk)}U_s^{(ijk)}+\frac{1}{2}H_{ls}^{(ijk)}\right)&=-\frac{1}{2}\phi^{\frac{1}{2}} V(\mu_{ijk})^{-\frac{1}{2}} \left(\frac{\partial^2 \mu_{ijk}}{\partial \eta_{ijk}^2}\right)^2x_{l(ijk)}x_{s(ijk)} \nonumber \\
     \sum_{l,s=1}^p\mathcal{K}^{ls}E(H_l^{(ijk)}U_s^{(ijk)}+\frac{1}{2}H_{ls}^{(ijk)})&=-\frac{1}{2}\phi^{\frac{1}{2}} V(\mu_{ijk})^{-\frac{1}{2}} \left(\frac{\partial^2 \mu_{ijk}}{\partial \eta_{ijk}^2}\right)^2x_{l(ijk)}x_{s(ijk)}\label{hlushls}
\end{align}
and therefore, from equations \eqref{primera_esperanza},  \eqref{hlushls} and \eqref{exprijk}:
\begin{equation}\label{eq8}
    E(R_{111}, R_{112}, \ldots, R_{nPL})=\frac{-1}{2\sqrt{\phi}}(\pmb{I}-\pmb{H})\pmb{J}\pmb{z}
\end{equation}
where
\begin{align*}
       \pmb{H}&=\pmb{W}^{\frac{1}{2}} \pmb{X(X^TWX)}^{\frac{1}{2}}\pmb{X}^T \pmb{W}^{\frac{1}{2}}\\
         \pmb{J}&=\pmb{D}\left( V(Y_{ijk}) \right)\pmb{D}\left( \frac{\partial \pmb{\mu}^2_i}{\partial^2\eta}\right)
\end{align*}
from equations \eqref{2rhlus}, \eqref{hlhs} y \eqref{rhls}:
\begin{align}
      -\sum_{l,s=1}^p \mathcal{K}^{ls}&E\left(2R_{ijk}H_l^{(ijk)}U_s^{(ijk)}+H_l^{(ijk)}H_s^{(ijk)} R_{ijk}H_{ls}^{(ijk)} \right)\nonumber\\
      &=\left[ -\phi w_{ijk}-\frac{\left(\frac{\partial V(\mu_{ijk})}{\partial \mu_{ijk}}\right)\left(\frac{\partial^2 \mu_{ijk}}{\partial \eta_{ijk}^2}\right)}{2 V(\mu_{ijk})}-\frac{1}{2}w_{ijk}\left(\frac{\partial^2 V(\mu_{ijk})}{\partial \mu^2_{ijk}}\right)\right]\frac{z_{ijk}}{\phi} \label{ecuacion15}
      \end{align}
 and from equations \eqref{sesgobeta} and \eqref{rhls} it follows that:
\begin{align}
      &2\sum_{l=1}^pB(\hat{\beta}_l)E(R_{ijk}H_{ls}^{(ijk)})\nonumber\\
      &=\frac{1}{2\phi}\frac{\left(\frac{\partial V(\mu_{ijk})}{\partial \mu_{ijk}}\right)\left(\frac{\partial^2 \mu_{ijk}}{\partial \eta_{ijk}^2}\right)}{V(\mu_{ijk})}\pmb{e}_{ijk} \pmb{Z}\pmb{D}(z_{ii})\pmb{D}\left( V(\mu_{ijk})^{-1} \left(\frac{\partial^2 \mu_{ijk}}{\partial \eta_{ijk}^2}\right) \left(\frac{\partial \mu_{ijk}}{\partial \eta_{ijk}}\right)\right)\pmb{1}\label{ecuacion 16}
\end{align}
therefore, from the results (\ref{ecuacion15}) and (\ref{ecuacion 16}) we have that:
\begin{equation}\label{eq9}
    \left[Var(R_{111}), Var(R_{112}), \ldots, Var(R_{nPL})\right]=\pmb{1}+\frac{1}{2\phi}\left(\pmb{QHJ-M}\right)\pmb{z}
\end{equation}
where $\pmb{1}$ is a vector of ones and
\begin{align*}
          \pmb{Q}&=\pmb{D}\left( V(Y_{ijk}) \right)^{\frac{1}{2}}\pmb{D}\left( \frac{\partial V(Y_{ijk})}{\partial\mu}\right)\\
          \pmb{M}&=\pmb{D}\left( V(Y_{ijk}) \right)^{\frac{1}{2}}\pmb{D}\left( \frac{\partial V(Y_{ijk})^2}{\partial^2\mu}+2\phi \pmb{D}(\pmb{W})+\frac{\frac{\partial V(Y_{ijk})}{\partial\mu}}{V(Y_{ijk})}\right)
\end{align*}
By theorem 2 of \citet{liang1986longitudinal} it is known that the GEE estimator of $\eta_{ijk}$ are consistent and unbiased, i.e., 
\begin{equation}\label{ecuacionz}
    \pmb{Z}\xrightarrow[n\to \infty]{p} \pmb{0}
\end{equation}
thus from \eqref{eq8} and \eqref{eq9}, we find that:
\begin{align}
    E(R_{ijk})&=O(n^{-1})\label{ceroon}\\
    Var(R_{ijk})&=1+O(n^{-1})\label{1varon}
\end{align}
and furthermore, by Section 3 of \cite{cordeiro2004Pearson} and equations \eqref{ceroon} and \eqref{1varon} it follows that:
\begin{equation}\label{normalidad}
    R_{ijk}\xrightarrow[n\to \infty]{d} N(0,1)
\end{equation}
Let
\begin{equation}\label{fff}
    \pmb{\Gamma}=\{\gamma_{ij}\}_{n\times n}
\end{equation}
be a matrix whose first column is $\frac{\pmb{1}}{\sqrt{n}}$ and the following columns are:
\begin{equation}
    g_{i-1}=\left(\frac{1}{\sqrt{(i-1)i}}, \ldots,\frac{1}{\sqrt{(i-1)i}},-\frac{i-1}{\sqrt{(i-1)i}}, 0, \ldots,0\right), \qquad i=2, \ldots, n
\end{equation}
\begin{equation}\nonumber
    \pmb{\Gamma}^T=\begin{pmatrix}
     \frac{\pmb{1}}{\sqrt{n}} &\vdots &      \pmb{G}
    \end{pmatrix}^T=\begin{pmatrix}
    \frac{1}{\sqrt{n}} & \frac{1}{\sqrt{n}}& \frac{1}{\sqrt{n}}& \frac{1}{\sqrt{n}} & \cdots&\frac{1}{\sqrt{n}}\\
    \frac{1}{\sqrt{2}} & -\frac{1}{\sqrt{2}}& 0& 0 &\cdots& 0\\
    \frac{1}{\sqrt{6}} & \frac{1}{\sqrt{6}}& -\frac{2}{\sqrt{6}}& 0 & \cdots&0\\
    \frac{1}{\sqrt{12}} & \frac{1}{\sqrt{12}}& \frac{1}{\sqrt{12}}& -\frac{3}{\sqrt{12}} & \cdots&0\\
    \vdots & \vdots& \vdots& \vdots & \ddots&\vdots\\
    \frac{1}{\sqrt{(n(n-1))}} & \frac{1}{\sqrt{n(n-1)}}& \frac{1}{\sqrt{n(n-1)}}& \frac{1}{\sqrt{n(n-1)}} & \cdots&-\frac{n-1}{\sqrt{n(n-1)}}\\
    \end{pmatrix}
\end{equation}
where
\begin{equation}\nonumber
      \pmb{G}^T=\begin{pmatrix}
       \frac{1}{\sqrt{2}} & -\frac{1}{\sqrt{2}}& 0& 0 &\cdots& 0\\
    \frac{1}{\sqrt{6}} & \frac{1}{\sqrt{6}}& -\frac{2}{\sqrt{6}}& 0 & \cdots&0\\
    \frac{1}{\sqrt{12}} & \frac{1}{\sqrt{12}}& \frac{1}{\sqrt{12}}& -\frac{3}{\sqrt{12}} & \cdots&0\\
    \vdots & \vdots& \vdots& \vdots & \ddots&\vdots\\
    \frac{1}{\sqrt{(n(n-1))}} & \frac{1}{\sqrt{n(n-1)}}& \frac{1}{\sqrt{n(n-1)}}& \frac{1}{\sqrt{n(n-1)}} & \cdots&-\frac{n-1}{\sqrt{n(n-1)}}\\
      \end{pmatrix}
\end{equation}
then the matrix $\pmb{\Gamma}$ is a Helmert matrix \citep{lancaster1965helmert} and therefore:
 \begin{equation}\label{porpGamma}
     \pmb{\Gamma}\pmb{\Gamma}^T=\pmb{I}_n, \qquad \pmb{1}_{n-1}^T \pmb{G}=0, \qquad \pmb{G}^T\pmb{G}=\pmb{I}_{n-1}-\frac{1}{n-1}\pmb{1}_{n-1}\pmb{1}_{n-1}^T
 \end{equation}
 
If $r_{ijk}$ is defined as the estimated Pearson residual of the $i$th experimental unit in the $j$th period and the $k$th observation, i.e., \[r_{ijk}=\hat{R}_{ijk}=\frac{Y_{ijk}-\hat{\mu}_{ijk}}{\hat{\phi}\sqrt{V(\hat{\mu}_{ijk})})}\]
and the matrix $\pmb{r}_i$ of residuals of the $i$th individual, where the first row has the $L$ Pearson residuals defined in Equation (\ref{ecPearson_obs}) corresponding to the first period, and the second row of the $L$ corresponding to the second period and so on until completing a matrix with $P$ rows and $L$ columns,  i.e.:
\begin{equation}\label{rrrr4}
      \pmb{r}_i=\begin{pmatrix}
    r_{i11}& r_{i21}& \cdots& r_{i1L}\\
    r_{i21}& r_{i22}& \cdots &r_{i2L}\\
    \vdots& \vdots& \ddots& \vdots\\
    r_{iP1}& r_{i2P}& \cdots &r_{iPL}\\
    \end{pmatrix}
    \end{equation}
By Equation \eqref{ceroon} and the correlation assumption given in Equation \eqref{ec105} it is true that:
\begin{align}
 &E( \pmb{r}_i)=\pmb{0}_{P\times L} \nonumber\\
    &Corr\left(Vec( \pmb{r}_i)\right)=\pmb{\Psi} \otimes \pmb{R}_1(\pmb{\alpha}_1)\label{resulri1}\\
    &Corr\left(Vec( \pmb{r}_i), Vec( \pmb{r}_{i'})\right)=\pmb{0}_{PL\times PL}\, \qquad i\neq i'\nonumber
\end{align}
And defining $\pmb{R}$ as:
\begin{equation}\nonumber
    \pmb{R}=(\pmb{r}_1,\ldots, \pmb{r}_n)_{P\times nL} \qquad
\end{equation}
and since $\pmb{\Gamma}$ is orthogonal, then $\pmb{\Gamma}\otimes \pmb{I}_L$ is also orthogonal. Thus \citep{srivastava2008models}: 
\begin{equation}
\pmb{R}(\pmb{\Gamma}\otimes \pmb{I}_L)=\begin{pmatrix}
\sqrt{n}\bar{\pmb{r}} &\vdots & \pmb{R}(\pmb{G}\otimes \pmb{I}_L)
\end{pmatrix} 
\end{equation}
and according to Equation (\ref{resulri1}), we get:
\begin{align*}
    \pmb{R}\left(\pmb{I}_n\otimes\pmb{ \Psi}^{-1} \right)\pmb{R}^T&=(\pmb{r}_1,\ldots, \pmb{r}_n)\left(\pmb{I}_n\otimes \pmb{ \Psi}^{-1} \right)(\pmb{r}_1,\ldots, \pmb{r}_n)^T\nonumber\\
    &=n\bar{\pmb{r}}\otimes  \pmb{ \Psi}^{-1} \pmb{R} + \pmb{R}(\pmb{G}\otimes \pmb{I}_L)\left(\pmb{I}_n\otimes \pmb{ \Psi}^{-1} \right)(\pmb{G}^T\otimes\pmb{I}_L )\pmb{R}^T\\
    &=n\bar{\pmb{r}}\otimes \pmb{ \Psi}^{-1} \pmb{R}+\pmb{Z}(\pmb{G}^T\otimes\pmb{I}_L )\pmb{Z}^T
\end{align*}
where $\pmb{Z}$ is:
\begin{align}
    &\pmb{Z}_{P\times(n-1)L}=(\pmb{Z}_1, \ldots, \pmb{Z}_{(n-1)})= \pmb{R}(\pmb{G}\otimes \pmb{I}_L)=(\pmb{r}_1,\ldots, \pmb{r}_n)(\pmb{G}\otimes \pmb{I}_L) 
\end{align}
with $\bar{\pmb{r}}$ is the matrix of the average residuals defined in Equation (\ref{rrrr4}) for each period, that is,
\begin{equation}\nonumber
    \bar{\pmb{r}}=\frac{1}{n}\sum_{i=1}^n \pmb{r}_i =\frac{1}{n}\sum_{i=1}^n\begin{pmatrix}
    r_{i11}& r_{i21}& \cdots& r_{i1L}\\
    r_{i21}& r_{i22}& \cdots &r_{i2L}\\
    \vdots& \vdots& \ddots& \vdots\\
    r_{iP1}& r_{i2P}& \cdots &r_{iPL}\\
    \end{pmatrix}
\end{equation}
and
\begin{align*}\nonumber
  &\pmb{Z}_1=\begin{pmatrix}
    \frac{1}{\sqrt{2}}r_{111}-\frac{1}{\sqrt{2}}r_{211} &
    \cdots & \frac{1}{\sqrt{2}}r_{11L}-\frac{1}{\sqrt{2}}r_{21L}\\
    \frac{1}{\sqrt{2}}r_{121}-\frac{1}{\sqrt{2}}r_{221} &
    \cdots & \frac{1}{\sqrt{2}}r_{12L}-\frac{1}{\sqrt{2}}r_{22L}\\
    \vdots &\ddots & \vdots\\
    \frac{1}{\sqrt{2}}r_{1P1}-\frac{1}{\sqrt{2}}r_{2P1} &
    \cdots & \frac{1}{\sqrt{2}}r_{1PL}-\frac{1}{\sqrt{2}}r_{2PL}
    \end{pmatrix}_{P\times L}\\
    \\
  &\pmb{Z}_2=\begin{pmatrix}
    \frac{1}{\sqrt{6}}\sum_{i=1}^2 r_{i11}-\frac{2}{\sqrt{6}}r_{311} &
    \cdots & \frac{1}{\sqrt{6}}\sum_{i=1}^2 r_{i1L}-\frac{2}{\sqrt{6}}r_{31L} \\
 \frac{1}{\sqrt{6}}\sum_{i=1}^2 r_{i21}-\frac{2}{\sqrt{6}}r_{321} &
    \cdots & \frac{1}{\sqrt{6}}\sum_{i=1}^2 r_{i2L}-\frac{2}{\sqrt{6}}r_{32L} \\
    \vdots & \vdots &\ddots & \vdots\\
    \frac{1}{\sqrt{6}}\sum_{i=1}^2 r_{iP1}-\frac{2}{\sqrt{6}}r_{3P1} &
    \cdots & \frac{1}{\sqrt{6}}\sum_{i=1}^2 r_{iPL}-\frac{2}{\sqrt{6}}r_{3PL} \\
    \end{pmatrix}_{P\times L}\\
    & \vdots\\
    &\pmb{Z}_{(n-1)}=\begin{pmatrix}
    \frac{1}{\sqrt{n(n-1)}}\sum_{i=1}^{n-1} r_{i11}-\frac{n-1}{\sqrt{n(n-1)}}r_{(n-1)11} &
    \cdots & \frac{1}{\sqrt{n(n-1)}}\sum_{i=1}^{n-1}r_{i1L}-\frac{n-1}{\sqrt{n(n-1)}}r_{(n-1)1L} \\
 \frac{1}{\sqrt{n(n-1)}}\sum_{i=1}^{n-1} r_{i21}-\frac{n-1}{\sqrt{n(n-1)}}r_{(n-1)21} &
    \cdots & \frac{1}{\sqrt{n(n-1)}}\sum_{i=1}^{n-1} r_{i2L}-\frac{n-1}{\sqrt{n(n-1)}}r_{(n-1)2L} \\
    \vdots & \vdots &\ddots & \vdots\\
    \frac{1}{\sqrt{n(n-1)}}\sum_{i=1}^{n-1} r_{iP1}-\frac{n-1}{\sqrt{n(n-1)}}r_{(n-1)P1} &
    \cdots & \frac{1}{\sqrt{n(n-1)}}\sum_{i=1}^{n-1} r_{iPL}-\frac{n-1}{\sqrt{n(n-1)}}r_{(n-1)PL} \\
    \end{pmatrix}_{P\times L}\\
\end{align*}
Now by the properties of the Pearson residuals we have that:
\begin{align}\nonumber
    E(\pmb{Z}_1)&=E(\pmb{Z}_2) = \cdots =E(\pmb{Z}_{n-1}) =\pmb{0}_{P\times L}
\end{align}
and by the properties given in Equation (\ref{porpGamma}) and because we assume that the experimental units are independent, that is, $Corr(r_{ijk}, r_{i'j'k'})=0$, for all $i\neq i'$ and that Equation \eqref{predic_eta} is true, then:
\begin{align*}
 & Corr(r_{ijk}, r_{i'j'k'})=0  \qquad \forall i\neq i'\\
    &Corr(r_{ijk}, r_{ij'k'})=Corr(r_{i'jk}, r_{i'j'k'})=Corr(r_{i'jk}, r_{i'j'k'}), \qquad \forall i\neq i'\\
    \qquad\\
    &Corr\left(\frac{1}{\sqrt{2}}r_{111}-\frac{1}{\sqrt{2}}r_{211},  \frac{1}{\sqrt{2}}r_{121}-\frac{1}{\sqrt{2}}r_{221}\right)=\frac{1}{2}Corr(r_{111}, r_{121})+\frac{1}{2}Cov(r_{211}, r_{221})\\
    &=Corr(r_{111}, r_{121})=Corr(r_{111}, r_{121})\\
    \qquad\\
    &Corr\left(\frac{1}{\sqrt{n(n-1)}}\sum_{i=1}^{n-1} r_{i11}-\frac{n-1}{\sqrt{n(n-1)}}r_{(n-1)11}, \frac{1}{\sqrt{n(n-1)}}\sum_{i=1}^{n-1} r_{i21}-\frac{n-1}{\sqrt{n(n-1)}}r_{(n-1)21}\right)\\
    &=Corr(r_{111}, r_{121})=Corr(r_{111}, r_{121})\\
\end{align*}
furthermore,
\begin{align}
 Var(Vec(\pmb{Z}_1))&=Var\left\{ \begin{pmatrix}
    \frac{1}{\sqrt{2}}r_{111}-\frac{1}{\sqrt{2}}r_{211} \\
    \frac{1}{\sqrt{2}}r_{121}-\frac{1}{\sqrt{2}}r_{221}\\
    \vdots\\
    \frac{1}{\sqrt{2}}r_{1P1}-\frac{1}{\sqrt{2}}r_{2P1}\\
    \frac{1}{\sqrt{2}}r_{112}-\frac{1}{\sqrt{2}}r_{212}\\
    \vdots\\
    \frac{1}{\sqrt{2}}r_{1PL}-\frac{1}{\sqrt{2}}r_{2PL}
    \end{pmatrix}_{PL\times 1} \right\} \nonumber\\
    &=\begin{pmatrix}
   1 & Corr(r_{111}, r_{121}) & \cdots & Corr(r_{111}, r_{1P1})\\
   Corr(r_{111}, r_{121}) & 1 & \cdots & Corr(r_{121}, r_{1P1})\\
   \vdots & \vdots & \ddots & \vdots \\
   Corr(r_{111}, r_{1P1}) & Corr(r_{121}, r_{1P1}) &\cdots &1
    \end{pmatrix}_{PL\times PL}\nonumber\\
    &=\pmb{ \Psi} \otimes \pmb{R}_1(\pmb{\alpha}_1))\nonumber
\end{align}
\begin{equation}\label{eczeta}
    Var(Vec(\pmb{Z}_1))=Var(Vec(\pmb{Z}_2))=\cdots = Var(Vec(\pmb{Z}_{(n-1)}))=\pmb{ \Psi} \otimes \pmb{R}_1(\pmb{\alpha}_1))
\end{equation}
By the central limit theorem, we get that:
\begin{align}
    Vec(\bar{\pmb{r}}) & \xrightarrow[n\to \infty]{d} N_{LP}(\pmb{0}, \pmb{ \Psi} \otimes \pmb{R}_1(\pmb{\alpha}_1))\label{rbarra}
\end{align}
By Equations \eqref{normalidad} and \eqref{eczeta} it follows that:
\begin{align}
  Vec(\pmb{Z}_j) &\xrightarrow[n\to \infty]{d} N_{LP}(\pmb{0},\pmb{ \Psi} \otimes \pmb{R}_1(\pmb{\alpha}_1))\label{aaata}
\end{align}
and by Equations \eqref{porpGamma}, \eqref{rbarra} y \eqref{aaata} that:
\begin{align}
    Cov&(Vec(\bar{\pmb{r}}),Vec(\pmb{Z}_j) )=\pmb{0}\nonumber \\
    Cov&(Vec(\pmb{Z}_i) ,Vec(\pmb{Z}_j) )=\pmb{0}\nonumber
\end{align}
and partitioning $\pmb{Z}_1$ as follows:
\begin{align*}
    \pmb{Z}_1&=\begin{pmatrix}
    \frac{1}{\sqrt{2}}r_{111}-\frac{1}{\sqrt{2}}r_{211} &
    \cdots & \frac{1}{\sqrt{2}}r_{11L}-\frac{1}{\sqrt{2}}r_{21L}\\
    \frac{1}{\sqrt{2}}r_{121}-\frac{1}{\sqrt{2}}r_{221} &
    \cdots & \frac{1}{\sqrt{2}}r_{12L}-\frac{1}{\sqrt{2}}r_{22L}\\
    \vdots &\ddots & \vdots\\
    \frac{1}{\sqrt{2}}r_{1P1}-\frac{1}{\sqrt{2}}r_{2P1} &
    \cdots & \frac{1}{\sqrt{2}}r_{1PL}-\frac{1}{\sqrt{2}}r_{2PL}
    \end{pmatrix}_{P\times L}\\
    &=\begin{pmatrix}
    \begin{pmatrix}
     \frac{1}{\sqrt{2}}r_{111}-\frac{1}{\sqrt{2}}r_{211}\\
     \frac{1}{\sqrt{2}}r_{121}-\frac{1}{\sqrt{2}}r_{221} \\
     \vdots\\
     \frac{1}{\sqrt{2}}r_{1P1}-\frac{1}{\sqrt{2}}r_{2P1} 
    \end{pmatrix}_{P\times 1}& \cdots &\begin{pmatrix}
    \frac{1}{\sqrt{2}}r_{11L}-\frac{1}{\sqrt{2}}r_{21L}\\
    \frac{1}{\sqrt{2}}r_{1PL}-\frac{1}{\sqrt{2}}r_{2PL}\\
    \vdots\\
    \frac{1}{\sqrt{2}}r_{1PL}-\frac{1}{\sqrt{2}}r_{2PL}
    \end{pmatrix}_{P\times 1}
    \end{pmatrix}\\
    &=(\pmb{z}_{(1)1}, \ldots, \pmb{z}_{(L)1}) \nonumber
\end{align*}
Then it is obtained that
\begin{equation}
    E(\pmb{z}_{(j)1}, \pmb{z}^T_{(j)1})=\Psi_{jj}\pmb{R}_1(\pmb{\alpha}_1) \qquad E(\pmb{Z}_1, \pmb{Z}^T_1)=(trace(\pmb{\Psi}))\pmb{R}_1(\pmb{\alpha}_1)\nonumber
\end{equation}
Similarly, it is found that:
\begin{align}
    &E(\pmb{z}_{(j)1}, \pmb{z}^T_{(j)1})=\Psi_{jj}\pmb{R}_1(\pmb{\alpha}_1)\nonumber\\
    &E(\pmb{z}_{(j)2}, \pmb{z}^T_{(j)2})=\Psi_{jj}\pmb{R}_1(\pmb{\alpha}_1)\nonumber\\
    \vdots\nonumber\\
 & E(\pmb{z}_{(j)(n-1)}, \pmb{z}^T_{(j)(n-1)})=\Psi_{jj}\pmb{R}_1(\pmb{\alpha}_1)\label{convzi}
 \end{align}
\begin{align}
     &   E(\pmb{Z}_1, \pmb{Z}^T_1)=(trace(\pmb{\Psi}))\pmb{R}_1(\pmb{\alpha}_1)\nonumber\\
     \vdots\nonumber\\
      &   E(\pmb{Z}_{(n-1)}, \pmb{Z}^T_{(n-q)})=(trace(\pmb{\Psi}))\pmb{R}_1(\pmb{\alpha}_1)
\end{align} 
Therefore since $\pmb{\Psi}_{jj}=1$, $\forall j=1,\ldots, P$, $trace(\pmb{\Psi})=1$ from Equation \eqref{convzi} we get:
\begin{equation}\label{yacasi}
   E\left( \pmb{R}_1(\pmb{\alpha}_1)^{\frac{1}{2}}\pmb{z}_{(j')k} \pmb{z}^T_{(j)k} \pmb{R}_1(\pmb{\alpha}_1)^{\frac{1}{2}}\right)=\pmb{\Psi}_{jj'}\pmb{I}_P, \qquad \forall k=1, \ldots, n-1
\end{equation}
\begin{equation}
   Cov[\pmb{R}_1(\pmb{\alpha}_1)^{\frac{1}{2}} \pmb{z}_{(k)j} \pmb{z}^T_{(i)j} \pmb{R}_1(\pmb{\alpha}_1)^{\frac{1}{2}}, \pmb{R}_1(\pmb{\alpha}_1)^{\frac{1}{2}} \pmb{z}_{(k)j'} \pmb{z}^T_{(i)j'} \pmb{R}_1(\pmb{\alpha}_1)^{\frac{1}{2}} ]=\pmb{0}
\end{equation}
By Theorem 2 in \cite{liang1986longitudinal} it is known that $\pmb{R}_1(\hat{\pmb{\alpha}}_1)$ is consistent and unbiased for $\pmb{R}_1(\pmb{ \alpha}_1)$, by Equations \eqref{yacasi}, \eqref{rbarra} and \eqref{aaata}, we  have that
\begin{equation}\label{107}
\hat{\psi}_{jj'}=\frac{1}{n} \sum_{i=1}^n tr\left(\pmb{R}_1(\pmb{\hat{\alpha}}_1)(\pmb{r}_{(j)i}-\bar{\pmb{r}}_{(j)})(\pmb{r}_{(j')i}-\bar{\pmb{r}}_{(j')})^T\right)
\end{equation}
is a consistent and asymptotically unbiased estimator for $\Psi_{jj'}$, which proves the theorem
\end{proof}
\end{appendix}

\end{document}